\documentclass[11pt]{article}
\usepackage{graphicx}
\usepackage{booktabs}
\usepackage{amsmath,amssymb,amsthm,fullpage}
\usepackage{xparse}
\usepackage{tikz}
\usetikzlibrary{arrows.meta,calc,positioning,decorations.pathreplacing}

\NewDocumentCommand{\interval}{O{0} m m O{} O{black} O{} O{}}{%
  \draw[very thick,#5] (#2,#1) -- (#3,#1);
  \draw[very thick,#5] (#2,#1+0.12) -- (#2,#1-0.12);
  \draw[very thick,#5] (#3,#1+0.12) -- (#3,#1-0.12);
  \IfNoValueF{#4}{\node[above=2pt] at ($(#2,#1)!0.5!(#3,#1)$) {#4};}
  \IfNoValueF{#6}{\node[below] at (#2,#1-0.12) {#6};}
  \IfNoValueF{#7}{\node[below] at (#3,#1-0.12) {#7};}
}

\NewDocumentCommand{\intervalthin}{O{0} m m O{} O{black} O{} O{}}{%
  \draw[#5] (#2,#1) -- (#3,#1);
  \draw[#5] (#2,#1+0.12) -- (#2,#1-0.12);
  \draw[#5] (#3,#1+0.12) -- (#3,#1-0.12);
  \IfNoValueF{#4}{\node[above=2pt] at ($(#2,#1)!0.5!(#3,#1)$) {#4};}
  \IfNoValueF{#6}{\node[below] at (#2,#1-0.12) {#6};}
  \IfNoValueF{#7}{\node[below] at (#3,#1-0.12) {#7};}
}

\newtheorem{theorem}{Theorem}
\newtheorem{lemma}{Lemma}

\title{Revoke vs. Restart in Unweighted Throughput Scheduling}
\author{
Changdao He\\
University of Toronto, Toronto, Canada\\
\texttt{changdao.he@mail.utoronto.ca}
}

\date{September 2025}

\begin{document}

\maketitle
\begin{abstract}
We study the unweighted throughput scheduling problem on a single machine in the 
\emph{preemption-revoke} model, where a running job may be aborted at any time, 
but all progress is permanently lost and the job cannot be restarted. Each job $J_i=(r_i,p_i,s_i)$ 
is defined by a release time $r_i$, a processing time $p_i$, and a slack $s_i$, 
and must start no later than $r_i+s_i$ to be feasible. 

We prove that no deterministic online algorithm can achieve a constant competitive ratio. 
The lower bound is established via an adversarial construction: starting from a 
three-job instance where $\textsf{ALG}$ completes at most one job while $\textsf{OPT}$ 
completes all three, we iteratively nest such constructions. By induction, for every $k\ge 3$, 
there exists an instance where $\textsf{ALG}$ completes at most one job, while 
$\textsf{OPT}$ completes at least $k$ jobs. Thus, the competitive ratio can be forced to 
$1/k$, and hence made arbitrarily close to zero. 

Our result stands in sharp contrast to the \emph{preemption-restart} model, where 
Hoogeveen et al.~\cite{Hoogeveen00} gave a deterministic $1/2$-competitive 
algorithm. 

\end{abstract}

\section{Introduction}

The interval and throughput scheduling problems arise in settings where time-sensitive jobs must be scheduled on a single machine. Each jobs $J$ has a release time $r$, processing time $p$, and deadline $d$, and must be completed within its available time window $[r, d)$. Equivalently, the deadline can be defined in terms of slackness, where slack $s=d-r-p$. 

The throughput scheduling objective is to maximize the number (or total weight) of jobs completed on time. A special case is interval scheduling, in which each job must start immediately upon release and run until its deadline, so every job has zero slack.

Two arrival models are commonly studied.
In the real-time model, jobs arrive in non-decreasing order of their release times, corresponding to how real systems encounter tasks in chronological order.
In contrast, the (any-order) online model allows jobs to arrive in arbitrary order, making it strictly more general and typically harder to achieve good competitive ratios. Depending on whether and how an algorithm may interrupt or abandon running jobs, four variants are distinguished:
\begin{enumerate}
    \item The non-preemptive model: once a job is started, it can neither be interrupted nor aborted.
    \item The preemption-revoke model: a job may be aborted at any time but is then permanently lost; this model is mainly studied in the (any-order) online setting.
    \item The preemption-restart model: a job may be aborted and later restarted from the beginning; this model has been primarily studied in real-time settings but could also be considered in online settings.
    \item The preemption-resume model: a job may be paused and later continued from where it left off; this model is mainly studied in real-time settings.
\end{enumerate}

\subsection{Related Work}

The problem of optimally maximizing the number of on-time jobs (equivalently, optimally minimizing the number of tardy jobs) 
originates with Lenstra et al.~\cite{Lenstra77}, who showed that the problem is NP-hard 
in the non-preemptive case. Lawler~\cite{Lawler90} gave the first polynomial-time algorithm for the 
preemption-resume model using dynamic programming, and Baptiste~\cite{Baptiste99} later improved the running 
time to $O(n^4)$, where $n$ denotes the number of jobs in the input. 

The real-time setting for online interval scheduling was introduced in the 1990s, where each job has a weight (or value) and corresponds to a fixed interval that must be accepted or rejected upon arrival. 
Lipton and Tomkins~\cite{Lipton} analyzed the case of two distinct processing times, designed a 2-competitive randomized algorithm, and showed that no algorithm can achieve a better competitive ratio.

In the preemption-revoke model, Woeginger~\cite{Woeginger} established that no online algorithm can achieve a finite competitive ratio when jobs may have arbitrary lengths and weights. For the special cases where job weights are given by C-benevolent (convex-increasing) or D-benevolent (monotone-decreasing) functions of their lengths, or where jobs have arbitrary weights but identical lengths,
he proposed a simple online heuristic (HEU) that achieves a competitive ratio of 4, and proved that no deterministic algorithm can do better. Faigle and Nawijn~\cite{FaigleNawijn} considered the unweighted setting in which all jobs have identical weight and may be assigned to one of $k$ identical machines, and presented an optimal (1-competitive) online algorithm for this case.
Fung et al.~\cite{Fung14} later proposed a 2-competitive barely random algorithm for equal-length jobs and for instances with C-benevolent or D-benevolent weight functions, and proved that the ratio 2 is optimal for all randomized algorithms on these instances.

Another line of research considers throughput scheduling with slack. For unweighted jobs, Chrobak et al.~\cite{Chrobak} analyzed the problem with equal-length jobs and demonstrated that randomization and restarts each improve performance. In the non-preemptive model, they obtained a barely random $5/3$-competitive algorithm. They showed that any barely random algorithm that chooses between two deterministic ones has ratio at least $3/2$, and that this lower bound can be improved to $8/5$ when the two are chosen with equal probability. In the preemption-restart model, they obtained a deterministic $3/2$-competitive algorithm and proved that the competitive ratio $3/2$ is optimal for deterministic algorithms. They also showed a lower bound of $6/5$ for randomized algorithms for this model. Fung et al.~\cite{Fung14} further considered arbitrary weights, equal-length jobs with restarts and designed a 3-competitive barely random algorithm.

In the preemption-resume model, Baruah et al.~\cite{Baruah94} showed that no deterministic online algorithm can achieve a constant competitive ratio.
Kalyanasundaram and Pruhs~\cite{Kalyanasundaram98} proved that randomization can help by presenting a barely random algorithm that achieves a constant competitive ratio, although the ratio is impractically large. Koren and Shasha~\cite{Koren95} considered the weighted version of the same problem and presented an online algorithm which achieves the best possible competitive ratio $(1+\sqrt{k})^2$, where $k$ is the ratio between the largest and smallest value densities among all jobs. Lucier et al.~\cite{Lucier13} designed constant-competitive algorithms under slackness assumptions, where each job’s time window is at least a constant factor larger than its processing time. Formally, if each job $J$ has release time $r$, processing time $p$, and deadline $d$, the slackness parameter $s$ is defined by $d-r\ge s \cdot p$. Their algorithm achieves a competitive ratio of at most $3+O((s-1)^{-2})$ for $1<s<2$ and $2+O(s^{-1/3})$ for $s\ge2$.

Related ideas have also appeared in online network scheduling. 
Garay et al.~\cite{Garay97} studied on-line call control on a line network, where calls must be accepted or rejected on arrival and may be revoked, and 
Garay et al.~\cite{Garay02} analyzed packet scheduling with interleaving, 
in which preempted packets can later resume transmission. Both of these network models belong to the more general (any-order) online setting, 
where requests may arrive in arbitrary order.

Borodin and Karavasilis~\cite{Borodin23} extended the interval scheduling problem to the (any-order) online model, where unweighted intervals may arrive in arbitrary order and previously accepted intervals may be revoked to accommodate new ones. They presented a simple deterministic algorithm that is $2k$-competitive when there are at most $k$ distinct interval lengths and proved that this bound is optimal for deterministic algorithms.

Most relevant to our work is the model of Hoogeveen et al.~\cite{Hoogeveen00}, 
who studied the preemption-restart model in the real-time setting and proved a tight 
$1/2$-competitive deterministic algorithm. 
In contrast, we analyze the preemption-revoke model, in which revoked jobs are permanently lost, 
and show that no deterministic online algorithm can achieve a constant competitive ratio.

\subsection{Our Contribution}
We show that no deterministic online algorithm can achieve a constant competitive ratio in this model. 
Starting from a three-job construction where $\textsf{ALG}$ completes at most one job while $\textsf{OPT}$ completes all three, 
we iteratively nest such constructions. By induction, we show that for every $k\ge 3$, there exists an instance where 
$\textsf{ALG}$ completes at most one job while $\textsf{OPT}$ completes at least $k$ jobs. 
Hence the competitive ratio can be forced to $1/k$, which tends to zero as $k$ grows.

\section{Preliminaries}
\label{sec:prelim}

\subsection{Problem Definition}
Jobs are triples $J_i=(r_i,p_i,s_i)$ with release time $r_i$, processing time $p_i>0$, and slack $s_i\ge 0$. The job is feasible only if it \emph{starts no later than} $r_i+s_i$; we call this requirement the \emph{latest-start constraint}. Let $d_i:=r_i+p_i+s_i$ be the end of the job's time window $[r_i,d_i)$. Notice that the job intervals are closed on the left and open on the right, this convention allows the machine to complete one job and immediately start another.

\subsection{The Preemption-Revoke Model}
In this model, the currently processed job may be revoked (aborted) at any time. 
If a job is revoked, all progress is permanently lost, and the job can neither be resumed nor restarted.
Thus, although an algorithm may revoke jobs during processing, the final solution it produces is non-preemptive.

\subsection{Notation and Convention on Competitive Ratios} In the following sections, we slightly abuse notation.  
We use $\textsf{ALG}$ to denote an arbitrary deterministic online algorithm and $\textsf{OPT}$ to denote an offline optimal schedule.  
When analyzing a specific instance, we also use $\textsf{ALG}$ and $\textsf{OPT}$ to represent the \emph{number of jobs completed} by the algorithm and the optimal schedule, respectively. Note that all jobs scheduled by $\textsf{OPT}$ are completed on time.

In the literature, competitive ratios are reported in two equivalent ways: either as numbers greater than~1 (e.g., ``the algorithm is 2-competitive'') or as fractions less than~1 (e.g., ``the algorithm achieves a competitive ratio of $1/2$''). 
Throughout this paper we follow the latter convention and express the ratio as
\[
\frac{\textsf{ALG}}{\textsf{OPT}} \le 1.
\]

\subsection{Notation Summary}

\begin{table}[ht]
\centering
\renewcommand{\arraystretch}{1.2}
\begin{tabular}{l l}
\toprule
\textbf{Symbol} & \textbf{Meaning} \\
\midrule
$J_i = (r_i, p_i, s_i)$ & Job $i$ with release time, processing time, and slack \\
$r_i$ & Release time of job $J_i$ \\
$p_i$ & Processing time (length) of job $J_i$ \\
$s_i$ & Slack of job $J_i$ \\
$d_i = r_i + p_i + s_i$ & Deadline or end of the job window \\
$[r_i, d_i)$ & Job window of $J_i$ \\
$[a_i,b_i)$ & Time interval during which $J_i$ is processed \\
\bottomrule
\end{tabular}
\caption{Summary of notation used throughout the paper.}
\label{tab:notation}
\end{table}

\section{A Warm-Up Lower Bound}

Before proving the full impossibility result, we start with a simple three-job instance that captures the essence of the preemption-revoke limitation. This warm-up provides intuition for the adversarial structure used later and already yields a lower bound of~$1/3$ on the competitive ratio of any deterministic online algorithm.

\begin{theorem}\label{thm:one-third}
For the unweighted throughput problem on one machine in the preemption-revoke model, no deterministic online algorithm is $>1/3$-competitive. In particular, for any deterministic online algorithm $\textsf{ALG}$ there is an input on which $\textsf{ALG}$ completes at most one job while $\textsf{OPT}$ completes three.
\end{theorem}

\noindent We first establish the following simple feasibility lemma used in the construction.

\begin{lemma}\label{lem:gap}
Let $J=(r,p,s)$ with $s\ge 2p$ and $d:=r+p+s$. Let $[x,y)\subset [r,d)$ be an interval of length $y-x\le p$. At least one of the side gaps, $[r,x)$ or $[y,d)$, has length $\ge p$. Moreover, either $y\le r+s$ so that $J$ can be scheduled in $[y, y+p)$, or else $x\ge r+p$ and $J$ can be scheduled in $[r, r+p)$; in either case $J$ can be scheduled entirely in one side gap while meeting its latest-start constraint. \emph{See Figure~\ref{fig:lemma-1}.}
\end{lemma}

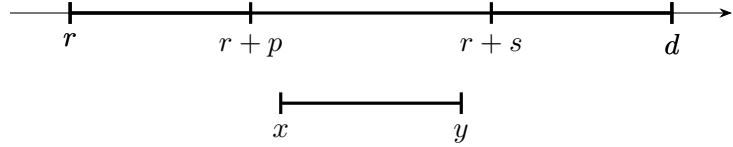
\begin{figure}[ht]
\centering
\begin{tikzpicture}[x=0.8cm,y=1.2cm,>={Stealth[length=2.5pt]}]
  \draw[-{Stealth}] (0,0) -- (12,0);
  \interval{1}{11}[][black][$r$][$d$]
  \interval[-1]{4.5}{7.5}[][black][$x$][$y$]
  \interval{8}{11}[][black][$r+s$][$d$]
  \interval{1}{4}[][black][$r$][$r+p$]
\end{tikzpicture}
\caption{Intervals $[r, d)$ and $[x, y)$ illustrating the side-gap condition.}
\label{fig:lemma-1}
\end{figure}

\begin{proof}
The total side-gap length is
\[
(x-r)+(d-y)=(d-r)-(y-x)=(p+s)-(y-x)\ge (p+s)-p=s\ge 2p,
\]
so at least one side gap has length $\ge p$. If $y\le r+s$, then starting job $J$ at $y$ is on time and would complete by $y+p\le r+s+p = d$. Otherwise $y>r+s$, recall that $y-x\le p$, then we have $p-r\ge (y-x)-r>r+s-x-r=s-x$. Hence $x-r>s-p\ge 2p - p=p$, so starting job $J$ at $r$ is on time and would complete by $r+p<x$.
\end{proof}

\begin{proof}[Proof of Theorem~\ref{thm:one-third}]
Fix any deterministic online algorithm $\textsf{ALG}$. We adversarially release jobs as follows.

\medskip\noindent\textbf{Step 1.}
Release $J_1=(r_1,p_1,s_1)$ with $s_1\ge2p_1$.

\medskip\noindent We branch on $\textsf{ALG}$'s response.

\medskip\noindent\textbf{Case 1 ($\textsf{ALG}$ ignores $J_1$).} If $\textsf{ALG}$ does not start $J_1$ before or on $r_1+s_1$, then $J_1$ expires and $\textsf{ALG}$ completes $0$ jobs while $\textsf{OPT}$ completes $1$, hence the competitive ratio is zero. This case already proves a lower bound; we therefore focus on the nontrivial case.

\medskip\noindent\textbf{Case 2 ($\textsf{ALG}$ runs $J_1$).} $\textsf{ALG}$ starts $J_1$ at some time $a_1$ where $r_1\le a_1\le r_1+s_1$ and would complete it at $b_1:=a_1+p_1\le d_1$.

\medskip\noindent\textbf{Step 2.}
Immediately after $a_1$, for some arbitrarily small $\varepsilon>0$, release a job $J_2=(r_2,p_2,s_2)$ with
\[
r_2:=a_1+\varepsilon,\qquad s_2\ge 2p_2,\qquad d_2:=r_2+p_2+s_2\le b_1.
\]
Thus $[r_2,d_2) \subset [a_1, b_1) \subset [r_1, d_1)$, and the entire window $[r_2,d_2)$ is \emph{strictly contained} in the processing interval $[a_1,b_1)$ of $J_1$. See Figure~\ref{fig:step-2}.

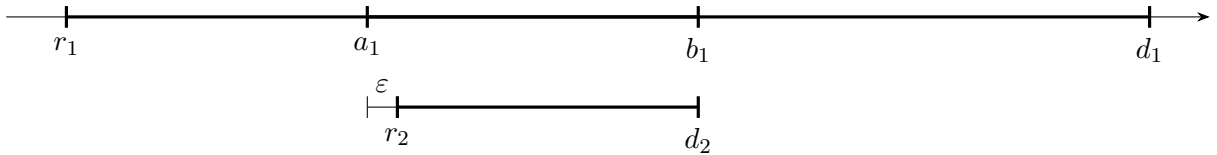
\begin{figure}[ht]
\centering
\begin{tikzpicture}[x=0.8cm,y=1.2cm,>={Stealth[length=2.5pt]}]
  \draw[-{Stealth}] (0,0) -- (20,0);
  \interval{1}{19}[][black][$r_1$][$d_1$]
  \interval{6}{11.5}[][black][$a_1$][$b_1$]
  \intervalthin[-1]{6}{6.5}[$\varepsilon$][black]
  \interval[-1]{6.5}{11.5}[][black][$r_2$][$d_2$]
\end{tikzpicture}
\caption{$J_2$ is released immediately after $a_1$.}
\label{fig:step-2}
\end{figure}

\medskip\noindent We again branch on $\textsf{ALG}$'s response.

\medskip\noindent\textbf{Case 2.1 ($\textsf{ALG}$ revokes $J_1$ and runs $J_2$).}
Suppose $\textsf{ALG}$ revokes $J_1$ and starts $J_2$ at time $a_2$ where $r_2\le a_2\le r_2+s_2$, and would complete it at $b_2:=a_2+p_2\le d_2$. For some arbitrarily small $\varepsilon>0$, release a job $J_3=(r_3,p_3,s_3)$ with
\[
r_3:=a_2+\varepsilon,\qquad s_3\ge 2p_3,\qquad d_3:=r_3+p_3+s_3\le b_2.
\]
Thus $[r_3,d_3)\subset [a_2,b_2)\subset [r_2,d_2)$, and the entire window $[r_3,d_3)$ is \emph{strictly contained} in the processing interval $[a_2,b_2)$ of $J_2$. See Figure~\ref{fig:case-2.1}.

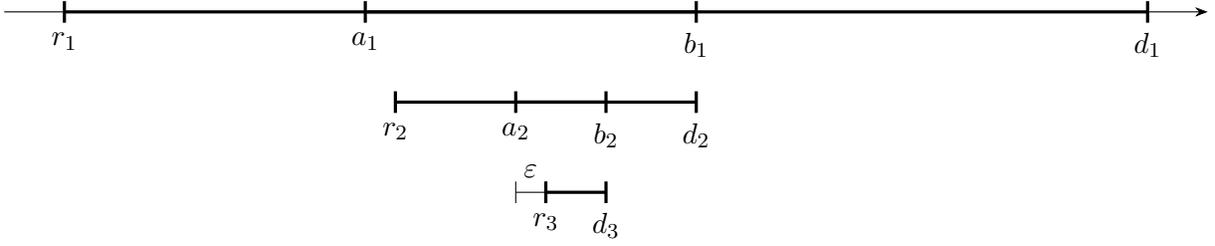
\begin{figure}[ht]
\centering
\begin{tikzpicture}[x=0.8cm,y=1.2cm,>={Stealth[length=2.5pt]}]
  \draw[-{Stealth}] (0,0) -- (20,0);
  \interval{1}{19}[][black][$r_1$][$d_1$]
  \interval{6}{11.5}[][black][$a_1$][$b_1$]
  \interval[-1]{6.5}{11.5}[][black][$r_2$][$d_2$]
  \interval[-1]{8.5}{10}[][black][$a_2$][$b_2$]
  \intervalthin[-2]{8.5}{9}[$\varepsilon$][black]
  \interval[-2]{9}{10}[][black][$r_3$][$d_3$]
\end{tikzpicture}
\caption{$J_3$ is released immediately after $a_2$.}
\label{fig:case-2.1}
\end{figure}

\medskip \emph{What $\textsf{ALG}$ can do.} It either (i) continues $J_2$ and ignores $J_3$, or (ii) revokes $J_2$ to runs $J_3$ in $[r_3, d_3)$. In either option it completes at most one job among $\{J_2, J_3\}$; together with $J_1$ revoked, we have $\textsf{ALG}\le 1$.

\emph{What $\textsf{OPT}$ can do.} It first schedules $J_3$ in $[r_3,d_3)$. By Lemma~\ref{lem:gap} applied to $J_2$ with the interval $[r_3,d_3)$ (of length $\le p_2$ and contained in $[r_2,d_2)$), there is a side gap that allows $J_2$ to be scheduled on time, either in $[r_2,r_3)$ or in $[d_3,d_2)$. Finally, since $[r_2,d_2)\subset [a_1,b_1)\subset [r_1,d_1)$ and $s_1\ge 2p_1$, Lemma~\ref{lem:gap} applied to $J_1$ with the interval $[r_2,d_2)$ (of length $p_2+s_2\le p_1$) gives a side gap for $J_1$ either in $[r_1,r_2)$ or in $[d_2,d_1)$. Hence $\textsf{OPT}=3$ and the ratio is at most $1/3$.

\medskip\noindent\textbf{Case 2.2 ($\textsf{ALG}$ ignores $J_2$).}
Assume $\textsf{ALG}$ continues processing $J_1$ and ignores $J_2$. For some arbitrarily small $\varepsilon>0$, release a job $J_4=(r_4,p_4,s_4)$ with
\[
r_4:=r_2+s_2+\varepsilon,\qquad s_4\ge 2p_4,\qquad d_4:=r_4+p_4+s_4\le d_2.
\]
Thus $[r_4,d_4)\subset [r_2+s_2,d_2)$. Recall that time $r_2+s_2$ is the latest feasible start for $J_2$. See Figure~\ref{fig:case-2.2}.

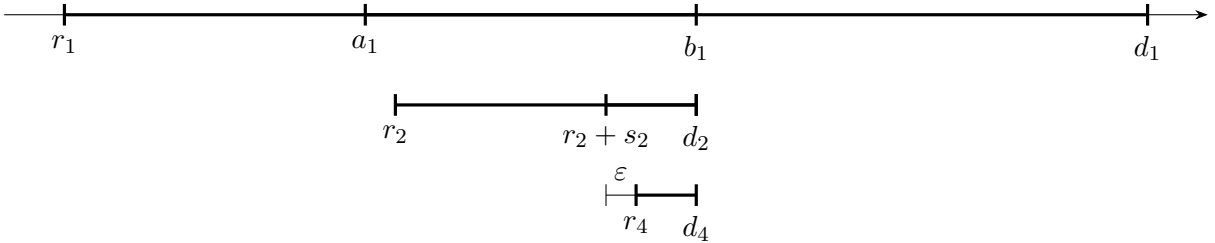
\begin{figure}[ht]
\centering
\begin{tikzpicture}[x=0.8cm,y=1.2cm,>={Stealth[length=2.5pt]}]
  \draw[-{Stealth}] (0,0) -- (20,0);
  \interval{1}{19}[][black][$r_1$][$d_1$]
  \interval{6}{11.5}[][black][$a_1$][$b_1$]
  \interval[-1]{6.5}{11.5}[][black][$r_2$][$d_2$]
  \interval[-1]{10}{11.5}[][black][$r_2+s_2$]
  \intervalthin[-2]{10}{10.5}[$\varepsilon$][black]
  \interval[-2]{10.5}{11.5}[][black][$r_4$][$d_4$]
\end{tikzpicture}
\caption{$J_4$ is released immediately after $r_2+s_2$.}
\label{fig:case-2.2}
\end{figure}

\medskip \emph{What $\textsf{ALG}$ can do.} Since $r_4=r_2+s_2+\varepsilon$, the job $J_2$ becomes infeasible for $\textsf{ALG}$ (its latest start has passed) after the release of $J_4$. So $\textsf{ALG}$ either (i) continues $J_1$ and ignores $J_4$, or (ii) revokes $J_1$ and runs $J_4$ in $[r_4,d_4)$. In either option $\textsf{ALG}$ completes at most one job among $\{J_1,J_4\}$; together with $J_2$ being infeasible, we have $\textsf{ALG}\le 1$.

\emph{What $\textsf{OPT}$ can do.} It first schedules $J_2$ starting at $r_2$; since $s_2\ge 2p_2$, we have $r_2+p_2<r_2+s_2<r_4$, so it would complete $J_2$ by time $r_4$. Next, it schedules $J_4$ in $[r_4,d_4)$. Finally, as before, apply Lemma~\ref{lem:gap} to $J_1$ with the interval $[r_2,d_2)$ (of length $p_2+s_2\le p_1$ and contained in $[a_1,b_1)$) to schedule $J_1$ either in $[r_1,r_2)$ or in $[d_2,d_1)$ while meeting its latest-start constraint. Hence $\textsf{OPT}=3$ and the ratio is at most $1/3$ in this subcase as well.

\medskip
Both \textbf{Case 2.1} and \textbf{Case 2.2} force the competitive ratio to be at most $1/3$, then \textbf{Case 2} has competitive ratio to be at most $1/3$. Combining \textbf{Case 1} and \textbf{Case 2}, either the competitive ratio is unbounded or at most $1/3$, proving the theorem.
\end{proof}

This warm-up illustrates that the \emph{preemption-revoke} model is fundamentally more difficult than the \emph{preemption-restart} model.
In the \emph{preemption-restart} model, the shortest-remaining-processing-time (SRPT) algorithm of Hoogeveen et al.~\cite{Hoogeveen00} achieves a tight competitive ratio of $1/2$. If the SRPT algorithm were applied to our lower-bound instance, according to the strategy of SRPT, a job is preempted only if a newly released job could be completed earlier. Then it would fall into Case 2.1, where it revokes $J_1$ and runs $J_2$ since $J_2$ would be completed earlier, and then it revokes $J_2$ and runs $J_3$ since $J_3$ would be completed earlier. Thus, SRPT also complete at most one of the three jobs, achieving the same $1/3$ ratio as any other deterministic algorithm.

In the next section, we extend this construction to show how the lower bound 
can be made arbitrarily small.

\section{No Constant Competitive Ratio}

We now generalize the $1/3$ construction to show that for every $k\ge 3$, there is an adversarial input on which any deterministic online algorithm $\textsf{ALG}$ completes at most one job, while $\textsf{OPT}$ completes at least $k$ jobs. Hence the competitive ratio can be forced to be $\le 1/k$, and therefore is arbitrarily close to $0$.

\begin{theorem}
For every $k \ge 3$, there exists an instance such that $\textsf{ALG}$ completes at most one job while $\textsf{OPT}$ completes at least $k$ jobs. 
Therefore, no deterministic online algorithm can achieve a constant competitive ratio in the preemption-revoke model.
\end{theorem}

\begin{proof}
We proceed by induction on $k$.

\medskip
\noindent\emph{Base case $k=3$.}
The construction from \textbf{Theorem~\ref{thm:one-third}} already shows the claim: 
in both \textbf{Case 2.1} and \textbf{Case 2.2}, $\textsf{ALG}$ completes at most one job, while $\textsf{OPT}$ completes all three jobs. 
Thus the statement holds for $k=3$.

\medskip
\noindent\emph{Inductive step.}
Suppose that for some $k \ge 3$ we have constructed an instance $\mathcal{I}_k$ consisting of jobs
$J_1,\ldots,J_k$ such that for all $i, J_i=(r_i,p_i,s_i)$ and $s_i\ge 2p_i$; which also satisfying:

\begin{itemize}
\item[(i)]
For each $i=1,\dots,k-1$ we ensure:
\[
[r_{i+1},d_{i+1}) \subset 
\begin{cases}
(a_i,b_i) & \text{if \textsf{ALG} starts $J_i$ at some $a_i\le r_i+s_i$, with $b_i:=a_i+p_i$,}\\[2mm]
[r_i+s_i,\,d_i) & \text{if \textsf{ALG} never starts $J_i$.}
\end{cases}
\]
In both subcases, $[r_{i+1},d_{i+1})\subset [r_i,d_i)$ and $\lvert [r_{i+1},d_{i+1})\rvert = p_{i+1}+s_{i+1} < p_i$;
\item[(ii)] Regardless of $\textsf{ALG}$'s choices, $\textsf{ALG}$ completes at most one job in $\mathcal{I}_k$;
\item[(iii)] $\textsf{OPT}$ can schedule all $k$ jobs in $\mathcal{I}_k$ (using Lemma~\ref{lem:gap} to place each $J_i$).
\end{itemize}

We create $\mathcal{I}_{k+1}$ by releasing one more job inside $J_k$, but we first branch on whether $\textsf{ALG}$ ever starts $J_k$.

\medskip\noindent\textbf{Case ($\textsf{ALG}$ runs $J_k$).}
$\textsf{ALG}$ starts $J_k$ at some time $a_k$ where $r_k\le a_k\le r_k+s_k$ and would complete it at $b_k:=a_k+p_k\le d_k$.
For some arbitrarily small $\varepsilon>0$, release a new job $J_{k+1}=(r_{k+1},p_{k+1},s_{k+1})$ with
\[
r_{k+1}:= a_k+\varepsilon\qquad
s_{k+1}\ge 2p_{k+1},\qquad d_{k+1}:=r_{k+1}+p_{k+1}+s_{k+1}\le b_k.
\]
Thus, $[r_{k+1},d_{k+1})\subset [a_k,b_k) \subset [r_k,d_k)$, and the entire window $[r_{k+1},d_{k+1})$ is \emph{strictly contained} in the processing interval $[a_k,b_k)$ of $J_k$. See Figure~\ref{fig:case-k+1.1}.

\begin{figure}[ht]
\centering
\begin{tikzpicture}[x=0.8cm,y=1.2cm,>={Stealth[length=2.5pt]}]
  \draw[-{Stealth}] (0,0) -- (20,0);
  \interval{1}{19}[][black][$r_k$][$d_k$]
  \interval{6}{11.5}[][black][$a_k$][$b_k$]
  \intervalthin[-1]{6}{6.5}[$\varepsilon$][black]
  \interval[-1]{6.5}{11.5}[][black][$r_{k+1}$][$d_{k+1}$]
\end{tikzpicture}
\caption{$J_{k+1}$ is released immediately after $a_k$.}
\label{fig:case-k+1.1}
\end{figure}
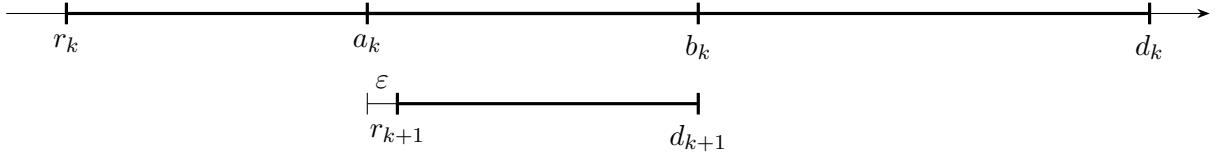

\medskip \emph{What $\textsf{ALG}$ can do.} It either (i) revokes $J_k$ and runs $J_{k+1}$ in $[r_{k+1}, d_{k+1})$ or (ii) continues $J_k$ and ignores $J_{k+1}$. In either option $\textsf{ALG}$ completes at most one job among $\{J_k,J_{k+1}\}$; together with all jobs $J_i$ for all $i<k$ being infeasible (note that by the construction, the machine runs $J_k$ implies that all earlier released jobs are either revoked or ignored), we have $\textsf{ALG}\le 1$.

\emph{What $\textsf{OPT}$ can do.} It first schedules $J_{k+1}$ in $[r_{k+1}, d_{k+1})$. By Lemma~\ref{lem:gap} applied to $J_k$ with interval $[r_{k+1},d_{k+1})$ (of length $\le p_k$ and contained in $[r_k, d_k)$), there is a side gap that allows $J_k$ to be scheduled on time, either in $[r_k,r_{k+1})$ or in $[d_{k+1},d_k)$. For all the jobs $J_i$ with $i<k$, schedule them the same way as scheduled for $\mathcal{I}_k$ (guaranteed by Lemma~\ref{lem:gap}). Hence $\textsf{OPT}$ = $k+1$.

\medskip\noindent\textbf{Case ($\textsf{ALG}$ ignores $J_k$).}
For some arbitrarily small $\varepsilon >0$, release a new job $J_{k+1}=(r_{k+1},p_{k+1},s_{k+1})$ with
\[
r_{k+1}:=r_k+s_k+\varepsilon,\qquad s_{k+1}\ge 2p_{k+1}, \qquad d_{k+1}:=r_{k+1}+p_{k+1}+s_{k+1}\le d_k.
\]
Then $[r_{k+1},d_{k+1})\subset [r_k+s_k,d_k)$. Recall that time $r_k+s_k$ is the latest feasible start for $J_k$. See Figure~\ref{fig:case-k+1.2}.

\begin{figure}[ht]
\centering
\begin{tikzpicture}[x=0.8cm,y=1.2cm,>={Stealth[length=2.5pt]}]
  \draw[-{Stealth}] (0,0) -- (20,0);
  \interval{1}{19}[][black][$r_k$][$d_k$]
  \interval{13.5}{19}[][black][$r_k+s_k$][$d_k$]
  \intervalthin[-1]{13.5}{14}[$\varepsilon$][black]
  \interval[-1]{14}{19}[][black][$r_{k+1}$][$d_{k+1}$]
\end{tikzpicture}
\caption{$J_{k+1}$ is released immediately after $r_k+s_k$.}
\label{fig:case-k+1.2}
\end{figure}
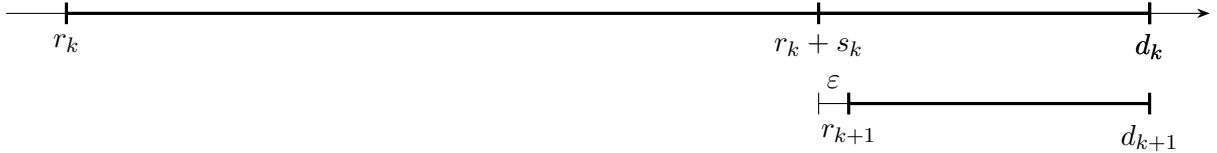

\medskip \emph{What $\textsf{ALG}$ can do.} Since $r_{k+1}=r_k+s_k+\varepsilon$, the job $J_k$ becomes infeasible for $\textsf{ALG}$ (its latest start has passed) after the release of $J_4$. So $\textsf{ALG}$ either (i) continues processing its current outer job (if any) and ignores $J_{k+1}$, or (ii) revokes the job that is currently being processed (if any) and runs $J_{k+1}$ in $[r_{k+1},d_{k+1})$. In either option $\textsf{ALG}$ completes at most one job among $\{J_1,J_2, ..., J_{k-1}, J_{k+1}\}$ (note that by the construction, the machine runs $J_i$ for some $i<k$ implies that all earlier released jobs are either revoked or ignored); together with $J_k$ being infeasible, we have $\textsf{ALG}\le 1$.

\emph{What $\textsf{OPT}$ can do.} It first schedules $J_k$ starting at $r_k$; since $s_k\ge 2p_k$, we have $r_k+p_k<r_k+s_k<r_{k+1}$, so it would complete $J_k$ by time $r_{k+1}$. Next, it schedules $J_{k+1}$ in $[r_{k+1},d_{k+1})$. For all the jobs $J_i$ with $i<k$, schedule them the same way as scheduled for $\mathcal{I}_k$ (guaranteed by Lemma~\ref{lem:gap}). Hence $\textsf{OPT}$ = $k+1$.

\medskip
By induction, for every $k\ge 3$, there exists an instance on which $\textsf{ALG}$ completes at most one job while $\textsf{OPT}$ schedules all $k$ jobs.
\end{proof}

\section{Conclusions and Discussion}

We established that no deterministic online algorithm can achieve a constant competitive ratio in the preemption-revoke model. 
Starting from a simple three-job construction, we showed how to iteratively extend the adversarial instance so that, for every $k \ge 3$, there exists an input sequence where the algorithm completes at most one job while the optimal offline schedule completes $k$ jobs. 
Hence, the competitive ratio can be forced to $1/k$, which approaches zero as $k$ grows.

Overall, the contrast with the preemption-restart model of Hoogeveen et al.~\cite{Hoogeveen00} highlights that allowing preempted jobs to restart preserves some competitiveness, whereas irrevocable revocation renders deterministic algorithms for the unweighted throughput scheduling problem fundamentally non-competitive, with the notable exception of the unweighted interval scheduling problem with revoking, where Faigle and Nawijn~\cite{FaigleNawijn} gave an optimal deterministic algorithm.

Beyond this lower bound, a natural direction is to identify restrictions on input instances that make a constant competitive ratio possible. In particular, if we consider time as discrete steps and allow only $1$ time unit of slack, or restrict jobs to those satisfying $s<2p$, can we obtain any constant competitive ratio for deterministic algorithms?
Another extension is to consider whether randomization can help: can a randomized algorithm, possibly with bounded memory (e.g., barely random), achieve a constant competitive ratio in the same model? 
It would also be interesting to investigate whether constant competitive ratios can be achieved for C-benevolent (or D-benevolent, respectively) weighted instances of the throughput problem, where the job weight is a convex-increasing (or monotone-decreasing, respectively) function of its processing time.
Of particular interest is the proportional-weight case $w_j=\alpha p_j$ for any constant $\alpha$, which is a special case of C-benevolent instances and corresponds to the uniform value-density model analyzed by Koren and Shasha~\cite{Koren95}.
Finally, one could ask whether a constant ratio $c(k)$ can be obtained in the revoke model when there are only $k$ distinct processing times, analogous to the result of Borodin and Karavasilis~\cite{Borodin23} for interval scheduling.

\section*{Acknowledgments}
I am deeply grateful to Professor Allan Borodin for his guidance and many helpful discussions that greatly improved this work.

\end{document}